\documentclass[12pt]{article}
\usepackage{bm}
\usepackage{mathtools}
\usepackage{url}
\usepackage{tikz}
\usetikzlibrary {arrows.meta}
\usepackage{comment}
\RequirePackage{amsthm,amsmath,amsfonts,amssymb}
\RequirePackage[numbers]{natbib}
\theoremstyle{plain}

\newtheorem{theorem}{Theorem}[section]
\newtheorem{lemma}[theorem]{Lemma}
\newtheorem{corollary}[theorem]{Corollary}

\newtheorem{proposition}[theorem]{Proposition}
\theoremstyle{plain}
\theoremstyle{remark}

\newtheorem{example}{Example}

\newtheorem*{remark}{Remark}
\newcommand{\norm}[1]{\left\lVert#1\right\rVert}
\newcommand{\Pb}{\mathbb{P}}
\newcommand*{\N}{\mathbb N}
\newcommand*{\R}{\mathbb R}
\newcommand*{\HH}{\mathcal{H}}
\newcommand*{\M}{\mathcal{M}}
\newcommand*{\I}{\mathcal{I}}
\newcommand*{\E}{\mathbb E}

\newcommand{\BB}[1]{{\color{black}{#1}}}
\begin{document}

%%%%%%%%%%%%%%%%%%%%%%%%%%%%%%%%%%%%%%%%%%%%%%
%%                                          %%
%% Enter the title of your article here     %%
%%                                          %%
%%%%%%%%%%%%%%%%%%%%%%%%%%%%%%%%%%%%%%%%%%%%%%
\title{Multivariate Representations of Univariate Marked Hawkes Processes}
\author{Louis Davis$^{1,2}$\hspace{.2cm},
    Conor Kresin$^1$,
   Boris Baeumer$^1$,    
    Ting Wang$^1$\\
    $^1$Department of Mathematics and Statistics, University of Otago,\\
    P.O. Box 56, Dunedin 9054, New Zealand\\
    $^2$E-mail: \url{ldavis2@stanford.edu}}

  \maketitle
\begin{abstract}
Univariate marked Hawkes processes are used to model a range of real-world phenomena including earthquake aftershock sequences, contagious disease spread, content diffusion on social media platforms, and order book dynamics. This paper illustrates a fundamental connection between univariate marked Hawkes processes and multivariate Hawkes processes. Exploiting this connection renders a framework that can be built upon for expressive and flexible inference on diverse data. Specifically, multivariate unmarked Hawkes representations are introduced as a tool to parameterize univariate marked Hawkes processes. We show that such multivariate representations can asymptotically approximate a large class of univariate marked Hawkes processes, are stationary given the approximated process is stationary, and that resultant conditional intensity parameters are identifiable. A simulation study demonstrates the efficacy of this approach, and provides heuristic bounds for error induced by the relatively larger parameter space of multivariate Hawkes processes.
 
\end{abstract}

\section{Introduction} The Hawkes point process model \citep{Hawkes1971} is widely used to model branching processes of interest ranging from contagious disease spread \cite{rizoiu2018sir,schoenberg2019recursive} and retaliatory crime \cite{mohler2011self} to earthquake aftershock sequences \cite{Ogata1988}, neural spike train data \cite{bonnet2023inference}, popularity of content on social media \cite{rizoiu2017expecting}, and high-frequency trading order book characteristics \cite{bacry2015hawkes}. Hawkes models and their variants are often employed to capture temporal clustering between points, roughly serving as an analog to autoregressive models for time series data that does not suffer from the modifiable areal unit problem \cite{daley2003introduction,fotheringham1991modifiable}. Hawkes models can accommodate marked point process data, which includes spatiotemporal data and event stream data as two categories of particular interest \cite{daley2003introduction,embrechts2018hawkes}, as well as multivariate data wherein different types of points can mutually excite or inhibit \cite{bremaud1996stability,liniger2009multivariate}. 

First order conditional intensity specifications of univariate marked Hawkes processes often assume mark-separability, an analytically and computationally necessary model mis-specification. For instance, the ground intensity and magnitude distribution of earthquake aftershocks are modelled as separable using the Epidemic Type Aftershock Sequence (ETAS) model \cite{Ogata1988} despite the fact that the ground intensity and magnitude distribution are postulated to be non-separable \cite{davis2024multidimensional, schoenberg2003multidimensional,Wang2012}. Although non-separable marked Hawkes processes provide a very general modelling framework, they are rarely implemented, as model specification is frequently challenging in practice \cite{zhuang2020detection}.

Parameters for conditional intensity functions governing point processes are often estimated via the method of maximum likelihood (MLE), as the resultant estimates are consistent and asymptotically normal \cite{ogata1978asymptotic}. The properties of MLE estimates of univariate Hawkes processes are discussed in detail in \cite{ozaki1979maximum}, and are often conjectured to hold for marked univariate and multivariate Hawkes processes \cite{bowsher2007modelling}. Reliance on MLE estimates increases the practical difficulty of fitting marked univariate Hawkes models as convergence to parameter estimates is slow \cite{veen2008estimation} and the computational expense of likelihood calculation is high \cite{reinhart2018review}. Less computationally costly estimation methods such as expectation maximization (EM) \cite{veen2008estimation}, spectral estimation \cite{cheysson2022spectral}, and Takacs Fiksel estimation \cite{kresin2023parametric} are available, but are relatively unstudied and infrequently implemented for estimation of marked Hawkes processes.

Precedent work demonstrating that a non-separable univariate marked Hawkes process can be expressed as a multivariate marked Hawkes process has shown promising results \cite{davis2024multidimensional}. However, the difficulty of computation and model specification of fitting multivariate \textit{marked} Hawkes processes is exacerbated relative to univariate marked Hawkes processes. In the context of simulation of multivariate Hawkes processes via thinning, the foundational connection between multivariate and univariate Hawkes processes is  well developed, see for instance Proposition 1 of \cite{ogata1981lewis}. Further, the more general theoretical framework of predictable projection of multivariate point processes is detailed in \cite{jacod1975multivariate}.

We propose partitioning the support of the mark distribution of a univariate marked Hawkes process and then modelling the process as an unmarked multivariate Hawkes process such that each component intensity corresponds to a distinct mark range. Our approach allows for both seperable and non-separable marked Hawkes processes to be captured flexibly via cross excitation in the multivariate Hawkes representation. The granularity of cross-excitation with respect to the mark distribution can be specified via the selected partitioning of the mark space, allowing for the true mark distribution to be unknown and free of modelling assumptions. Further, our approach retains the benefits of modelling the univariate marked process parametrically as the estimated parameters are interpretable. 

Two existing methods utilise multivariate Hawkes representations of non-separable univariate marked Hawkes processes. The first existing approach is a class of Hawkes-like models, here termed \textit{Neural Hawkes}. Neural Hawkes models are neural networks built on architectures ranging from continuous long short-term memory (cLSTM) \cite{mei2017neural} to transformer \cite{zuo2020transformer}. Such models are commonly applied in the case of event stream data; such data is naturally modelled as a non-separable  univariate marked Hawkes process with a discrete mark space (wherein each mark is a \textit{type} of point). Neural Hawkes models are flexible unmarked multivariate models that can parameterize cross excitement, model nonstationarity, nonlinearity, and even negative background rates. However, such models do not yield interpretable parameters and the finite sample and asymptotic statistical properties of Neural Hawkes estimates are yet to be developed \cite{schafer2006recurrent,shchur2021neural}. In contrast, our method relies on the well-trodden theory surrounding multivariate Hawkes processes \cite{liniger2009multivariate}. 

A second approach, closer to our proposed method, partitions the mark space of a (possibly non-separable) univariate marked Hawkes process whereupon the second order characteristics of the resultant multivariate Hawkes representation are shown to be the solution to a system of Fredholm integrals estimable via the Wiener-Hopf method \cite{Bacry2016First}. This methodology has been applied to order book data \cite{Rambaldi2017}, and allows for both exciting and inhibitory effects non-parametrically. Unlike our approach, this methodology relies on completely nonparametric estimation of the kernel matrix, and like the cLSTM, results in non-interpretable estimates. Further, the multivariate Hawkes representation used in our method is shown to approximate a univariate marked Hawkes process arbitrarily well; both above-mentioned existing methods fail to provide any such guarantee.

\subsection{Our contribution}
We show that for any given realization, there exists a multivariate unmarked Hawkes process representation that can approximate the conditional intensity function of a univariate marked Hawkes process so that the $L^1$ difference between the two decreases as the number of components increases. Furthermore, we demonstrate that such a multivariate representation is stationary given stationarity of the targeted univariate marked Hawkes process. Finally, we demonstrate that the parameters governing the multivariate Hawkes representation are identifiable. This method excels relative to other existing options with respect to the cases where (1) the true mark distribution is unknown or non-stationary, and (2) if the mark distribution and ground intensity of the univariate marked process are non-separable or require parametric model specifications that are computationally intractable. Cross excitement parameterized by a multivariate Hawkes process allows for the ground intensity and mark distribution interaction present in a marked univariate Hawkes process to be flexibly modelled in a computationally tractable setting. These results are demonstrated in a very general setting wherein we only assume the existence of a positive mark density, continuity of mark-variable productivity function, and integrability of the positive kernel function.

The remainder of the paper is structured as follows: Section \ref{sec:notation} contains necessary notation and preliminaries, Section \ref{sec:results} contains the above-described results, Section \ref{Sec:Simulation} contains a simulation study validating the results described in Section \ref{sec:results}, and Section \ref{sec:conclusion} concludes and briefly discusses future work.

\section{Introduction to Point Processes and Hawkes Processes}\label{sec:notation}
A point process \BB{$N$} is a measurable mapping from a filtered probability space $(\Omega,\mathcal{H},\mathcal{P})$ onto $\mathcal{N}$, the set of $\mathbb{Z}^+$-valued random measures (counting measures) on a complete separable metric space (CSMS) \cite{daley2003introduction}. We will restrict our attention to point processes that are boundedly finite, \textit{i.e.} processes having only a finite number of points inside any bounded set.

Such a point process is assumed to be adapted to the filtration $\{\mathcal{H}_t\}_{t\geq 0}$ containing all information on the process $N$ at all locations and all times up to and including time $t$. A process is $\mathcal{H}$-predictable if it is adapted to the filtration generated by the left continuous processes $\mathcal{H}_{t}$.  Intuitively, $\mathcal{H}_{t}$ represents the history of a process up to, but not including time $t$. A rigorous definition of $\mathcal{H}_{t}$ can be found in \cite{daley2008introduction}. For a Borel set $X$, we say that $N(X)$ is the number of the points in the set $X$, and $N(T)\vcentcolon=N\left([0,T)\right)$ is the number of points that have arrived by time $T$. Assuming it exists, the $\mathcal{H}$-{\sl conditional intensity} $\lambda$ of $N$ is an integrable, non-negative, $\mathcal{H}$-predictable process
such that  
\begin{equation}\label{eqn:condintensity}
    \lambda(t,M|\HH_t) = \lim_{h, \delta \downarrow 0} \frac{\E[N\left([t,t+h) \times \mathbb{B}(M,\delta)\right) | \mathcal{H}_{t}]}{h \mu\left( \mathbb{B}(M,\delta)\right)}.
\end{equation}
where $\mathbb{B}(M,\delta)$ is a ball centered at $M$ with radius $\delta$, $\mu$ is the measure on $\M$ and $\mathcal{H}_{t}$ represents the history of the process $N$ up to but not including time $t$.

A point process is {\sl simple} if with probability one, all the points are distinct. Since the conditional intensity $\lambda$ uniquely determines the finite-dimensional distributions of any simple point process 
(Proposition 7.2.IV of \cite{daley2003introduction}), one typically models a simple point process by specifying a model for $\lambda$. A point process is {\sl stationary} if  for all bounded Borel subsets of the real line $\mathcal{B}_r$ and times $t\in \R$, $\{N(\mathcal{B}_r+t)\}_{r=1}^\infty \perp t$, which is to say the joint distribution of the specified model has a structure which is invariant over shifts in time.

A multivariate Hawkes process \cite{Hawkes1971} is a simple point process $\{N_i(t)\}_{i=1}^k$ or equivalently $\bm N(t)$, which can be represented as a branching process. The process is defined as the sum of the $K$ component intensities, each expressed as a Stieltjes integral   
\begin{equation}
    \lambda^{(i)}(t|\mathcal{H}_t)=\lambda_{0i}+\sum_{j=1}^K\int_{0}^t g_{ij}(t-u)dN_j(u).
\end{equation}
When $K=1$, such a process is referred to as a univariate Hawkes process. The properties of general multivariate processes \cite{bremaud1981point} and nonlinear Hawkes processes \cite{bremaud1996stability} are well developed, but the scope of this paper is restricted to linear Hawkes processes. Such processes can be viewed as Poisson cluster processes made up of immigrant events of type $i$, arriving as a Poisson process with rate $\lambda_{0i}$ and offspring events corresponding to each immigrant event given a history (e.g. \cite{hawkes1973cluster}). Offspring events are those excited by previous arrivals in the history of the process, and we say that event $t_i$ is a first generation offspring of event $t_j$, if $t_j$ is the parent of $t_i$. Hence, Hawkes processes are stationary if $\rho\left(\int_0^\infty \bm g(u) d\bm N(u)\right)<1$ where $\rho(\cdot)$ denotes the spectral radius. In this case, the mean cluster size is finite (Lemma 6.3.II of \cite{daley2003introduction}). 

The conditional intensity function of a marked univariate Hawkes process takes the parametric form 
\begin{equation}\label{eqn:markedunivariate}
    \lambda(t,M|\HH_t)=\lambda_0(M)+\int_{0}^t g(t-u,M)dN(u).
\end{equation}
If the factorization $g(t-u,M)=f(M)g^\prime(t-u)$, $\lambda_0(M)=\Lambda f(M)$ exists, where $f(\cdot)$ is the mark density, $g^\prime(\cdot)$ is the kernel function, and $\Lambda$ is the Poisson immigrant rate, then the process characterized in Equation \eqref{eqn:markedunivariate} is mark separable.

Within this article we are interested in the properties of a multivariate Hawkes process on the space $[0,T] \times \mathcal{M}$ where $\M$ is a compact CSMS equipped with metric $d$ and measure $\mu$ such that $\mu(\M)<\infty$. When comparing our multivariate process to a marked univariate process we use the same history for each process. Therefore, our convergence theorems can be viewed as pointwise deterministic. It is likely that our process converges to marked univariate Hawkes processes in the Skorokhod topology, however this is outside of the scope of the paper and not pertinent for practical purposes. 

We now introduce our multivariate Hawkes representation which will be used to approximate a non-separable univariate marked Hawkes process. Specifically, we define the conditional intensity function, dependent on the parameter vector $\theta$ to be
\begin{equation}\label{eq:generalmodel}
    \lambda_\theta(t,M|\HH_t)=\sum_{i=1}^K \chi_{A_i}(M)\left(\lambda_{0i}+\sum_{j=1}^K\sum_{t_{\ell,j}:t_{\ell,j}<t}\alpha_{ij}g_{ij}(t-t_{\ell,j};\beta_{ij})\right)
\end{equation}
where $\chi_{A}(M)$ is the indicator function for the measurable set $A\subset \M$, $\lambda_{0i}$ is the background rate parameter, $t_{\ell,j}$ is the $\ell^{th}$ event of type $j$ (i.e. $M_{\ell,j} \in A_j$) and $g_{ij}(t;\beta_{ij})$ is a positive density dependent on only one parameter $\beta_{ij}$, which we now shorten to $g_{ij}(t)$ to ease notation. For any fixed $K$ the parameter vector is $\theta=\{\{\lambda_{0i}\}, \{\alpha_{ij}\},\{\beta_{ij}\} \}_{{1\leq i,j\leq K}}$, hence there are $2K^2+K$ parameters, thus $\theta \in \Theta \subset \R^{K(2K+1)}_{+}$.

Such a representation is constructed on a set of measurable sets $\left\{A_j\right\}_{1\leq j\leq K}$ that partition $\M$ such that
\begin{itemize}
    \item[P1.] $K\in \N$ is finite,
    \item[P2.] $A_i \cap A_j = \emptyset$ if $i \neq j$,
    \item[P3.] $\bigcup_{i=1}^K A_i=\M$, and
    \item[P4.] $\mu(A_i)>0$ for every $i=1,2,\dots K$.    
\end{itemize}
The purpose of properties P1-P3 is to construct a finite partition of $\M$. P4 ensures that the probability an event has a mark taking a value in $A_i$ is positive for every set $A_i$ and so that the density has a finite mean value on the set $A_i$. For example, $\M$ could be a compact subset of $\R^d$, in which case we would use the usual Euclidean metric and Lebesgue measure, and the point process would then be spatial-temporal. Spatial-temporal Hawkes processes have many applications such as earthquakes \cite{Ogata1998,Ogata2006}, crime \cite{mohler2011self} or terrorism \cite{Jun2024}. Alternatively, if $\M\subset \mathbb Z^d$ one should use the discrete metric and counting measure. Hawkes processes with integer marks have been previously applied to social media event stream data \cite{rizoiu2017hawkes}, graphical representations of multivariate Hawkes processes \cite{embrechts2018hawkes} and a wide variety of financial data applications \cite{embrechts2011multivariate}.

\section{Results}\label{sec:results}

In this section we present the main results of our article in full detail. Specifically, we demonstrate the existence of a multivariate representation that approximates a univariate marked Hawkes process in the $L^1$ sense. Following this, under slightly stricter assumptions, we also demonstrate that the multivariate representation is stationary if the target process is also stationary. We finally demonstrate that the parameters of the multivariate representation are identifiable.

\subsection{The special case of Mark-Separable Point Processes}\label{sec:MarkSeparableCase}

In this subsection we first demonstrate that there exists a multivariate representation of a compound Poisson process. We demonstrate that the MLE of the multivariate representation, without excitation, are exactly those of a histogram estimator. This specific case is included to build intuition as to what our multivariate representation does. To do so, we require Proposition 14.3.II (c) of \cite{daley2008introduction}.

\begin{proposition}[Proposition 14.3.II (c) of \cite{daley2008introduction}]\label{prop:GandMDVJ}
    Given a predictable conditional intensity $\lambda(t,M|\HH_t)$, the associated ground process has predictable conditional intensity
\begin{equation}\label{eq:groundProcess}
\lambda_g(t|\HH_t)=\int_{\mathcal M}\lambda(t,M|\HH_t)d\mu(M)
  \end{equation}
and the associated mark distribution with density
  \begin{equation}\label{eq:MarkDensity}
      f(M|\HH_t)=\frac{\lambda(t,M|\HH_t)}{\lambda_g(t|\HH_t)}.  \end{equation}
  \end{proposition}

We use Proposition \ref{prop:GandMDVJ} in the following example to demonstrate the use of our multivariate representation.

\begin{example}\label{thm:PoissProc}

Consider a compound Poisson process with conditional intensity $\Lambda_0f(M)$ where $\Lambda_0>0$ is the arrival rate of events and $f(M)$ is the density for the independent and identically distributed marks taking values in a compact subset of $\R^d$. Suppose $\left\{A_j\right\}_{1\leq j\leq K}$ satisfy properties P1-P4 and that $\alpha_{ij}=0$ for every $i,j$. In this case, the MLE of the multivariate representation (Equation \eqref{eq:generalmodel}) are the same as the estimates given by a histogram estimator. 

 To demonstrate this, let $\{(t_i,M_i)\}_{1\leq i \leq N(T)}\subset [0,T]\times \M$ be a realisation of a compound Poisson process. By Proposition 7.3.III of \cite{daley2003introduction} the log-likelihood for the multivariate representation of the compound Poisson process is
\begin{align}
    \mathcal{L}(\bm \theta)=&\sum_{j=1}^{N(T)}\log\left(\sum_{i=1}^K \chi_{A_i}(M_j)\lambda_{0i} \right)-\int_{\mathcal M}\int_0^{T}\sum_{i=1}^K \chi_{A_i}(M)\lambda_{0i}dtd\mu(M)\nonumber \\
    =&\sum_{j=1}^{N(T)}\log(\lambda_{0k_j})-\sum_{i=1}^K \mu(A_i)\lambda_{0i}T=\sum_{i=1}^K\left(N_{i}(T)\log(\lambda_{0i})-\lambda_{0i}\mu([0,T]\times A_i) \right) \label{eq:PoissLogL}
\end{align}
where $k_j$ is the type of event $j$ (i.e. $M_j \in A_{k_j}$), $N_i(T)$ is the number of events that have arrived by time $T$ such that $M \in A_i$ and $\mu$ is the Lebesgue measure on $[0,T]\times \M \subset \R^+ \times \R^d$. Since each summand on the right hand side of Equation \eqref{eq:PoissLogL} is independent, $\lambda_{0i}$ can be estimated by maximising each summand separately. Define
\[f(\lambda_{0i})=N_i(T)\log(\lambda_{0i})-\mu([0,T]\times A_i)\lambda_{0i}.\]
It directly follows that 
\[f'(\lambda_{0i})=\frac{N_i(T)}{\lambda_{0i}}-\mu([0,T]\times A_i)=0 \implies \hat \lambda_{0i}=\frac{N_i(T)}{\mu([0,T]\times A_i)}=\frac{N_i(T)}{T\mu(A_i)}\]
which is clearly a maximum by taking the second derivative of $f$. Furthermore, $\hat \lambda_{0i}$ is exactly the MLE for a Poisson process observed on $[0,T]\times A_i$, therefore inheriting consistency and asymptotic normality, e.g. Example 1 of \cite{ogata1978asymptotic}. 

We now examine the induced ground process for the estimated multivariate representation using Proposition \ref{prop:GandMDVJ}. By Equation \eqref{eq:groundProcess},
\[ \hat \lambda_g(t|\HH_t)= \sum_{i=1}^K \mu(A_i)\hat\lambda_{0i}=\sum_{i=1}^K \frac{\mu(A_i)N_i(T)}{T\mu(A_i)}=\frac{N(T)}{T}\]
which is exactly the MLE for the rate of an unmarked Poisson process on the set $[0,T]$. Hence, $\hat \lambda_g(t|\HH_t)$ converges in probability to $\Lambda$. Consider the estimated density function computed using Equation \eqref{eq:MarkDensity}
\begin{align*}
    \hat f(M)=\frac{\sum_{i=1}^K \frac{N_i(T)\chi_{A_i}(M)}{T\mu(A_i)}}{\sum_{i=1}^K \frac{N_i(T)}{T}}    =\frac{T}{N(T)}\sum_{i=1}^K \frac{N_i(T)\chi_{A_i}(M)}{T\mu(A_i)}
\end{align*}
\begin{equation}\label{eq:HistEstimator}
    \implies \hat f(M)= \frac{1}{N(T)}\sum_{i=1}^K\frac{N_i(T)\chi_{A_i}(M)}{\mu(A_i)}.
\end{equation}
It is apparent that Equation \eqref{eq:HistEstimator} is the histogram estimator for the mark density of an i.i.d. sample of size $N(T)$. Hence, if the cells $A_i$ are selected correctly this estimator of $f$ is consistent.
\end{example}

\begin{remark}
The assumption of setting $\alpha_{ij}=0$ for every $i,j$ may seem unjustified theoretically. However, this assumption could be justified in practice, by using a one-sided Kolmogorov-Smirnov test for the uniformity of $\{t_i\}_{1\leq i \leq N(T)}$ on $[0,T]$ which may provide insufficient evidence to suggest temporal clustering of events.
\end{remark}

As a result of Example \ref{thm:PoissProc}, the analytic properties of a histogram estimator are automatically inherited for our multivariate representation. When $\mathcal M =[0,1]$ we present the main results of \cite{Wasserman2006} Section 6.2 which suggests how to optimally select the sets $A_i$ so that the histogram estimator converges to the true density. First, the optimal histogram estimator $\hat f$ to the true density $f$ minimises the mean integrated square error, $\text{MISE}(\hat f,f)$. 

In this case we should select $A_i$ as half-open intervals of the form
$A_i=\left[\frac{i-1}{K},\frac{i}{K} \right),$ each of uniform width $h=1/K$; clearly these sets satisfy properties P1-P4. By Theorem 6.11 of \cite{Wasserman2006}, under mild assumptions on $f$, the width of $A_i$ that minimises the MISE is $\mathcal{O}\left(N(T)^{-1/3}\right)$ in which case $\text{MISE}\left(\hat f,f\right)=\mathcal{O}\left(N(T)^{-2/3}\right)$. Therefore, for optimal statistical efficiency in the case of real valued marks, one should select $K=\lceil N(T)^{-1/3} \rceil$ and uniform bins $A_i$ of width $K^{-1}$. For more information see Section 3 of \cite{izenman1991review} and references therein.

While analytic results for the MLE given the multivariate representation of a compound Poisson processes can be derived and shown to be equivalent to known estimators, history dependent processes are of primary interest. The following discussion formalises how a multivariate Hawkes representation can approximate a broad class of non-separable marked Hawkes processes. Approximation of separable marked Hawkes processes is shown as a special case, as such processes are subsumed within the larger class of non-separable marked Hawkes processes. 

\subsection{Non-separable Point Processes}

For the following discussion we use a fixed realisation for each process, so that the conditional intensity functions have discontinuities at the same instants in time. Underlying our discussion is the result that any measurable function can be approximated by a simple function, which is the intuition for Theorems \ref{prop:SeparableHPL1} and \ref{thm:Stationarity}.

This section is devoted to showing that our multivariate representation is able to represent a large class of non-separable univariate marked Hawkes processes. Suppose that $(\I_1,\mathcal{F}_1,\Pb_1)$ and $(\I_2,\mathcal{F}_2,\Pb_2)$ are probability spaces and that
\[f_1:\I_1\to L^1(\M), \quad f_2:\I_2 \to L^1(\M), \quad \xi:\I_2\to C(\M).\]
Specifically, $f_1,f_2$ are density valued random variables and $\xi$ is a continuous function valued random variable for all of the mark space $\M$. Specifically, for every $\omega_1\in \I_1, \omega_2\in \I_2$, $f_1(M;\omega_1)$ and $f_2(M;\omega_2)$ are densities on $\M$ with respect to the same reference measure $\mu$ and $\xi(M;\omega_2)$ is a continuous real-valued function known as the mark variable productivity, a continuous analog to those discussed in \cite{chiang2022hawkes, paik2022nonparametric}. We next consider the general form of a conditional intensity characteristic of a marked linear Hawkes process

\begin{multline}\label{eq:separableHPINT}
\lambda_{\text{HP}}(t,M|\HH_t)=\Lambda\int_{\I_1}f_1(M;\omega)d\Pb_1(\omega)+\\
\int_{\I_2}f_2(M;\omega)\sum_{i:t_i<t}g(t-t_i;M_i)\xi(M_i;\omega)d\Pb_2(\omega)
\end{multline}
where $g(t;M)$ is the time varying kernel dependent on the parameter $\beta(M)$, such that for every $M\in \M$, $\int_0^\infty g(s;M)ds=1$.

The incremental process of arrivals for a given realisation of a non-separable Hawkes process can be thought of in stages. Since $f_1,f_2$ are densities pointwise on $\I_1,\I_2$ respectively, their expected values are as well ($f_1,f_2$ are non-negative so use Tonelli's theorem to integrate over $\M$ first). We now apply Equation \eqref{eq:groundProcess} and find that points arrive with intensity
\[\lambda_g(t|\HH_t)=\Lambda+\sum_{i:t_i<t}g(t-t_i)\int_{\I_2}\xi(M_i;\omega)d\Pb(\omega),\]
which is exactly the ground process of a marked Hawkes process. Once a point arrives, at time $t$, it is an immigrant (in which case set $\tau=1$) with probability \[p=\frac{\Lambda}{\Lambda+\sum_{i:t_i<t}g(t-t_i)\int_{\I_2}\xi(M_i;\omega)d\Pb_2(\omega)},\] or an offspring event (in which case $\tau=2$) with probability $1-p$. Then the mark of that point is simulated as a realisation of the random variable with density $f_\tau(M;\omega_\tau)$. Heuristically, for a given point in a given realisation, the mark value depends on three rolls of dependent dice. 

We now formalise the result, in Theorem \ref{prop:SeparableHPL1}, that there exists a parameter vector for Equation \eqref{eq:generalmodel} such that $\lambda_\theta(t,M|\HH_t)$ can approximate $\lambda_{\text{HP}}(t,M|\HH_t)$ well in the $L^1$ sense.

\begin{theorem}[Existence]\label{prop:SeparableHPL1}
Consider a non-separable marked Hawkes process with intensity given by Equation \eqref{eq:separableHPINT} where $f_1,f_2$ are densities with respect to the measure $\mu$ and $\xi$ is a real-valued continuous function. Moreover, suppose that $f_1,f_2,\xi$ are all equicontinuous on $\M$ with respect to $\I_1$ and $\I_2$. Furthermore, assume that $g(t;\beta)$ is uniformly continuous with respect to $\beta$ and that $\beta(M)$ is continuous in $\M$. Then for every $T>0$, fixed realisation, and $\epsilon>0$ there exists a $K \in \N$, a parameter vector $\tilde \theta_K$, and set of sets $\left\{A_i\right\}_{1 \leq i \leq K}$ that satisfies properties P1-P4 such that, 

\[\norm{\lambda_{\text{HP}}(t,M|\HH_t)-\lambda_{\tilde\theta_K}(t,M|\HH_t)}_{L^1([0,T]\times \mathcal M)}<\epsilon.\]
\end{theorem}

\begin{proof}[Proof Sketch]
The sketch of the proof is as follows. We first show \[\bar f_1(M)=\int_{\I_1}f_1(M;\omega_1)d\Pb_1(\omega_1), \quad \overline{f_2\xi}(M,M')=\int_{\I_2}f_2(M;\omega)\xi(M';\omega)\Pb_2(\omega), \]
are uniformly continuous on $\M$. We then use compactness to decompose $\M$ into sets $A_i$ so that we can approximate $\bar f_1$ and $\overline{f_2\xi}$ by a simple function arbitrarily well. $L^1$ convergence is a directly follows. A complete proof is presented in Appendix \ref{sec:L1to0Proof}. 
\end{proof}

We note that Theorem \ref{prop:SeparableHPL1} could be extended component-wise to a marked multivariate Hawkes process with $J$ component intensities. Furthermore, in specific cases where $\lambda_{\text{HP}}(t,M|\HH_t)$ varies deterministically in time it is possible to select a different parameter vector $\tilde \theta_K$ to still approximate $\lambda_{\text{HP}}(t,M|\HH_t)$, if the representation is selected wisely. Example \ref{ex:1} demonstrates this for a mark-separable renewal process.

\begin{example}\label{ex:1}
    Consider a renewal mark-separable Hawkes process, where the immigrant events arrive as a Weibull renewal process. Hence
    \[\lambda_{\text{HP}}(t,M|\HH_t)=f(M)\left(ba^b(t-t')^{b-1}+\sum_{i:t_i<t}g(t-t_i)\xi(M_i) \right)\]
    where $a,b$ are parameters and $t'$ is the most recent immigrant event. Then to approximate $\lambda_{\text{HP}}(t,M|\HH_t)$ by our multivariate representation the parameter vector is then of the form $\theta=\{\{a_{i}\}, \{b_{i}\} \{\alpha_{ij}\},\{\beta_{ij}\} \}_{{1\leq i,j\leq K}}$. Pick $\left\{A_i\right\}_{1 \leq i \leq K}$, $\alpha_{ij}$ and $\beta_{ij}$ as in the proof of Theorem \ref{prop:SeparableHPL1}. Additionally, assume that an immigrant event resets the renewal rate independent of its mark and set $b_i=b$ and $a_i=a\left(\tilde f_i\right)^{\frac{1}{b_i}}$. In this case, using $\lambda_{0i}(t)=b_ia_i^{b_i}(t-t')^{b_{i}-1}$ will approximate $f(M)ba^b(t-t')^{b-1}$ in the $L^1$ sense.
\end{example}

We now show that the multivariate representation commensurate with the parameter vector such that $\lambda_\theta(t,M)$ approximates a non-separable Hawkes process on the domain $[0,T] \times \M$ is stationary if the target Hawkes process is stationary, and for each event its mark-variable productivity is deterministic.

\begin{theorem}[Stationarity]\label{thm:Stationarity}
    Suppose the assumptions of Theorem \ref{prop:SeparableHPL1} are satisfied and the Hawkes process $\lambda_{\text{HP}}(t,M|\HH_t)$ is stationary when $\xi$ is independent of $\omega$. Then $\lambda_{\tilde\theta_K}(t,M|\HH_t)$, constructed in Theorem \ref{prop:SeparableHPL1}, is also stationary.
\end{theorem}

\begin{proof}[Proof Sketch]
    The stationarity condition for the non-separable Hawkes processes is that the expected cluster size is finite. To prove this, we first demonstrate that the expected number of offspring from an immigrant is finite. We then show that the total number of offspring from a non-immigrant is also finite. In this case, each independent branching process with one immigrant in the centre, a.s. has finitely many events. The condition for this is that    
    \[I=\int_\M \int_{\I_2}f_2(M;\omega)\xi(M) d\Pb_2(\omega) d\mu(M)<1.\]
    
    Using the parameter vector $\tilde \theta_K$ and sets $\{A_i\}_{1\leq i \leq K}$, as constructed in the proof of Theorem \ref{prop:SeparableHPL1}, we approximate $I$ by the integral of a simple function. We then compute the expected number of first generation offspring for our multivariate representation and show that it is strictly bounded above by $1$ demonstrating that the mean cluster size is finite a.s. concluding our argument. The full proof is presented in Appendix \ref{sec:proofstationarity}.
\end{proof}

\begin{remark}
    We conjecture Theorem \ref{thm:Stationarity} to be true for a stochastic mark variable productivity. This remains conjectured as it is challenging to compute the expected number of first generation offspring from a multivariate representation. To do so, Lemma \ref{lem:StatDensity} must be modified so that the mark distribution is fully history dependent. This in and of itself is not an issue, but in this case $\alpha_{ij}$ cannot be factorised into a term dependent on $i$ and a term dependent on $j$. The computation of the expected number of first generation offspring from a non-immigrant event relies on the factorisation of $\alpha_{ij}$, and therefore in this context is intractable. However, in most, if not all, applications $\xi$ is assumed to be deterministic so this slight loss of generality is not a major concern.
\end{remark}

While an expected result, Theorem \ref{thm:Stationarity} demonstrates that our multivariate representation has the potential to be used in practice. Of primary importance, Theorem \ref{thm:Stationarity} implies the existence of the multivariate representation $\lambda_\theta$ via Khinchin's Existence Theorem (Proposition 3.3.I of \cite{daley2003introduction}), as well as orderliness of $\lambda_\theta$, a result upon which the proof of Theorem \ref{thm:TheoremIdentifiability} is implicitly reliant \cite{daley2003introduction}. Also of the utmost importance is that the preservation of stationarity, and hence non-explosivity almost surely (a.s.), is crucial for frequentist forecasting. Predictive performance is often performed by repeatedly simulating realisations of the process over some subset of the time domain \cite{vere1998probabilities,zhuang2011next}. Therefore, if the data generating process is finite, then once the multivariate representation approximates the data generating process well enough it too is stationary facilitating the construction of frequentist forecasts. Finally, the results of Theorem \ref{prop:SeparableHPL1} imply that our representation can approximate the true data generating process arbitrarily well.

We again note that mark-separable Hawkes processes are a special case of Equation \eqref{eq:separableHPINT}, which we now discuss in the following example.

\begin{example}\label{ex:SepHawkesProcess}
Consider the mark separable Hawkes process on $[0,T]\times [M_0,M_m]$ induced by the conditional intensity function
\begin{equation}\label{eq:SepHP}
    \lambda^*(t,M|\HH_t)\vcentcolon= f(M)\left(\Lambda+\sum_{i:t_i<t}g(t-t_i)\xi(M_i)\right)
\end{equation}
where $f$ is a continuous density. It is a direct result of Theorems \ref{prop:SeparableHPL1} and \ref{thm:Stationarity} that our multivariate representation can approximate $\lambda^*$ arbitrarily well, and is stationary if the process induced by $\lambda^*$ is stationary. In the case of the mark separable Hawkes process this is that
\[I=\int_\M f(M)\xi(M)d\mu(M)<1.\]

We now present a corollary, which says that the induced mark density of the multivariate representation is a simple function approximation of the mark density $f$.
\begin{corollary}\label{cor:SepMarkDensApp}
    Suppose the target process is a mark separable Hawkes process with continuous mark density and mark variable productivity given by Equation \eqref{eq:SepHP}. Then the mark density induced by $\lambda_{\tilde \theta_K}(t,M|\HH_t)$ is 
    \[\tilde f(M)=\sum_{i=1}^K \bar f_i\chi_{A_i}(M), \quad \text{ where } \bar f_i=\frac{1}{\mu(A_i)}\int_{A_i}f(M)d\mu(M).\]
\end{corollary}
\begin{proof}
Since the conditions of Theorem \ref{prop:SeparableHPL1} are satisfied we compute the induced mark density, $f'$, using Equation \eqref{eq:MarkDensity} in Proposition \ref{prop:GandMDVJ} with $\tilde \theta_K$ and the set of sets $\{A_i\}$ for $ 1\leq i \leq K$ that exist as a consequence of the theorem. Specifically, $\lambda_{0i}=\Lambda \bar f_i$, $\beta_{ij}=\beta$ and $\alpha_{ij}=\bar f_i\xi_j$.  Then 
    \begin{align*}
     \tilde f(M)=&\frac{\lambda_{\tilde \theta_K}(t,M|\HH_t)}{\int_\M \lambda_{\tilde \theta_K}(t,M|\HH_t)d\mu(M)}\\=&\frac{\sum_{i=1}^K \chi_{A_i}(M)\left(\lambda_{0i}+\sum_{j=1}^K\sum_{t_{\ell,j}:t_{\ell,j}<t}\alpha_{ij}g_{ij}(t-t_{\ell,j})\right)}{\sum_{i=1}^K \mu(A_i)\left(\lambda_{0i}+\sum_{j=1}^K\sum_{t_{\ell,j}:t_{\ell,j}<t}\alpha_{ij}g_{ij}(t-t_{\ell,j})\right)}\\
     =&\frac{\sum_{i=1}^K \chi_{A_i}(M)\left(\Lambda \bar f_i+\sum_{j=1}^K\sum_{t_{\ell,j}:t_{\ell,j}<t}\xi_j \bar f_i g(t-t_{\ell,j})\right)}{\sum_{i=1}^K \mu(A_i)\left(\Lambda \bar f_i+\sum_{j=1}^K\sum_{t_{\ell,j}:t_{\ell,j}<t}\xi_j \bar f_ig(t-t_{\ell,j})\right)}\\
      =&\frac{\sum_{i=1}^K \bar f_i \chi_{A_i}(M)\left(\Lambda +\sum_{j=1}^K\sum_{t_{\ell,j}:t_{\ell,j}<t}\xi_j g(t-t_{\ell,j})\right)}{\sum_{i=1}^K \bar f_i\mu(A_i)\left(\Lambda+\sum_{j=1}^K\sum_{t_{\ell,j}:t_{\ell,j}<t}\xi_j g(t-t_{\ell,j})\right)}
      =\sum_{i=1}^K \bar f_i\chi_{A_i}(M)
    \end{align*}
where we obtained the finally equality since $\sum_{i=1}^K \bar f_i\mu(A_i)= 1$. In particular the marks of the multivariate representation are i.i.d.\ for every $K \in \N$, which is as expected considering the target process has i.i.d.\ marks. 
\end{proof}

Corollary \ref{cor:SepMarkDensApp} implies that when modelling potentially separable processes, the magnitude density can be estimated in a non-parametric way assuming $\hat \theta_T \xrightarrow[N(T)\to \infty]{p} \tilde \theta_K$ as a consequence of the continuous mapping theorem \cite{Mann1943Stochastic} since $\lambda_\theta(t,M|\HH_t)$ is continuous with respect to $\theta$. Therefore, once asymptotic distributions for the maximum likelihood estimates of multivariate processes are established analytically, it may be possible to develop tests for mark separability and compare it to existing methods such as \cite{schoenberg2004testing}.

\end{example}

\subsection{Identifiability}\label{sec:identifiable}

We conclude this section by demonstrating that the parameters of the multivariate representation are identifiable. Previous studies have partitioned the mark space in a similar way to us, however they have all performed non-parametric estimation e.g. \cite{Bacry2016First,Rambaldi2017}. Our multivariate representation is parametric, and because information is lost by turning continuous marks (for example if $\M \subset \R^d$) into discrete types, parameter identifiability is not necessarily trivial. Fortunately, given the below assumptions, the parameters of our multivariate representation defined by Equation \eqref{eq:generalmodel}, are identifiable.

\begin{itemize}
    \item[A1.] The point process with intensity $\lambda_\theta(t,M|\HH_t)$ is stationary.
    \item[A2.] The intensity $\lambda_\theta(t,M|\HH_t)$ is positive for every $(t,M) \in [0,T] \times \mathcal{M}$. 
    \item[A3.] $g_{ij}(t)$ is linearly independent of $g'_{ij}(t)$ when $\beta_{ij}\neq \beta'_{ij}$. 
    \item[A4.] In the time interval $[0,T]$ there exists an event such that $(t_k,M_k)$ such that $M_k \in A_j$ for every $j$ and that $t_1>0$.
\end{itemize}

Under assumption A1 the process is finite since it is both simple and stationary e.g. Proposition 3.3.V of \cite{daley2003introduction}. In this case the process is orderly, and $\exists \epsilon_i>0$, $\forall i \in \N$, such that $t_{i+1}-t_i>\epsilon_i$ a.s..

Assumption A2 can be relaxed slightly so that the intensity can be 0 for strict subsets of $[0,T] \times \M$. However in doing so, we would be required to define restart times, similar to \cite{bonnet2023inference}, ensuring no event of type $i$ arrives at time $t$ when $\lambda_\theta(t,M|\HH_t)=0$ for every $M \in A_i$. We leave the close extension of identifiability for parametric estimation of nonlinear multivariate Hawkes processes characterising univariate marked processes for future work.

The purpose of assumption A3 is to ensure that for a unique set of parameter values the conditional intensity function changes in time uniquely. There are several types of parametric functions that satisfy assumption A3 and have been used in univariate Hawkes process models. Examples include exponential functions $g_{ij}(t)=\exp\{-\beta_{ij}t\}$ \cite{Hawkes1971,kresin2022comparison}, power law functions for a fixed $p>0$, $g_{ij}(t)=(t+c_{ij})^{-p}$ \cite{Ogata1988,Ogata1998}, or $g_{ij}(t)=te^{-\beta_{ij}t}$ \cite{Ogata1982,VEREJONESD1982}.

Finally, assumption A4 ensures that there exists a non-empty subset of $[0,T]\times \mathcal{M}$ where the intensity depends on every parameter. Heuristically, if there was never an event of type $1$, i.e. $M_k \not \in A_1$ for every $k$, then $\lambda_{\theta}(t,M|\HH_t)$ will never depend on $\alpha_{11}$. In which case the probability structure of the process is independent of $\alpha_{11}$, and hence this parameter should not be identifiable.

\begin{theorem}[Identifiability]\label{thm:TheoremIdentifiability}
Fix $K \in \N$ to be arbitrary. Assume that the set of sets $\left\{A_j\right\}_{1\leq j\leq K}$ satisfy properties P1-P4 and that assumptions A1-A4 are satisfied. Then parameters of the multivariate representation, specified by Equation \eqref{eq:generalmodel}, are identifiable.
\end{theorem}

\begin{proof}[Proof Sketch]
    Because conditional intensity functions uniquely define the finite dimensional distributions of the point process \cite{daley2003introduction}, to show identifiability we must demonstrate that for almost every $(t,M) \in [0,T]\times \mathcal M$ that $\lambda_\theta(t,M|\HH_t)=\lambda_{\theta'}(t,M|\HH_t)  \iff \theta=\theta'$. It is obvious that if $\theta=\theta'$ then $\lambda_\theta(t,M|\HH_t)=\lambda_{\theta'}(t,M|\HH_t)$.

Our argument for the other direction relies on partitioning $[0,T]\times \mathcal M$ into cells of the form $(t_{k},t_{k+1}]\times A_i$ (e.g. Figure \ref{fig:IdentifiabilityDiscretisationMB}) where $t_k$ is the $k^{th}$ event. We then prove identifiability of the model in the cells $(t_{1},t_{2}]\times A_i$, where $t_1$ is the first event. We then induct over the event types to prove the model is identifiable in the next cell where the conditional intensity function depends on new parameters. 
The full proof of Theorem \ref{thm:TheoremIdentifiability} is presented in Appendix \ref{Sec:IdentifiabilityProof}.
\begin{figure}[h!]
    \centering
    \begin{tikzpicture}
    \draw[step=2cm,gray,very thin] (-2.001,-2) grid (5.9,6);
    \draw (-2,6.2) node {0};
    \draw (-2.4,5.8) node {$M_0$};
    \draw (0,6.2) node {$t_1$};
    \draw (2,6.2) node {$t_2$};
    \draw (4,6.2) node {$t_3$};
    \draw (-1,-1) node {$A_4$};
    \draw (-1,1) node {$A_3$};
    \draw (-1,3) node {$A_2$};
    \draw (-1,5) node {$A_1$};
    \draw (1,-1) node {$A_4$};
    \draw (1,1) node {$A_3$};
    \draw (1,3) node {$A_2$};
    \draw (1,5) node {$A_1$};
    \draw (3,-1) node {$A_4$};
    \draw (3,1) node {$A_3$};
    \draw (3,3) node {$A_2$};
    \draw (3,5) node {$A_1$};
    \draw (5,-1) node {$A_4$};
    \draw (5,1) node {$A_3$};
    \draw (5,3) node {$A_2$};
    \draw (5,5) node {$A_1$};
     \draw[gray,very thin,-Stealth] (5.9,6) -- (5.99,6);
     \draw[gray,very thin,-Stealth] (5.9,4) -- (5.99,4);
     \draw[gray,very thin,-Stealth] (5.9,2) -- (5.99,2);
     \draw[gray,very thin,-Stealth] (5.9,0) -- (5.99,0);
     \draw[gray,very thin,-Stealth] (-2,-2) -- (-2,-2.01);
     \draw[gray,very thin,-Stealth] (0,-2) -- (0,-2.01);
     \draw[gray,very thin,-Stealth] (2,-2) -- (2,-2.01);
     \draw[gray,very thin,-Stealth] (4,-2) -- (4,-2.01);
\end{tikzpicture}
    \caption{Partitioning of $[0,T]\times \M$}
    \label{fig:IdentifiabilityDiscretisationMB}
\end{figure}
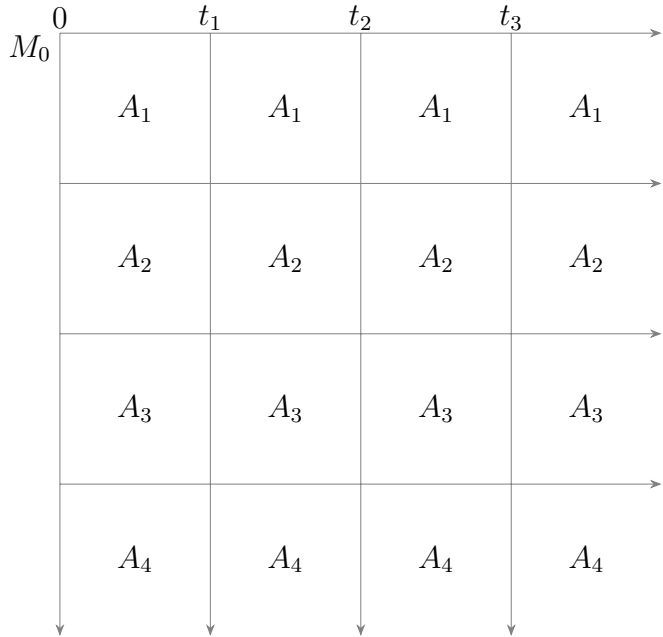
\end{proof}
\section{Simulation study}\label{Sec:Simulation}
The following study provides numerical validation for Theorem \ref{prop:SeparableHPL1}. We observe convergence of simulated realisations of the univariate process observed on increasingly large time windows to unmarked multivariate processes with an increasing number of component intensities. We conclude this study with discussion of the increased model complexity induced by the (possibly) larger number of parameters necessitated by a multivariate representation relative to a univariate representation. 
\subsection{Simulation procedure} We consider a stationary univariate marked Hawkes process with a constant background rate and exponential kernel function, \textit{i.e.} \begin{equation}\label{eqn:sim_uni}
    \lambda_{\text{HP}}(t|\mathcal{H}_t)=\lambda_0+\sum_{i:t_i<t} \alpha\exp(-\beta(t-t_i))
\end{equation} which is a special instance of Equation \eqref{eq:separableHPINT} where $\xi(M)=\alpha$ and $f_1=f_2$ where both are constant in $\omega$. In this case, the time varying kernel of the target Hawkes process, $\alpha \exp(-\beta t)$, is not a density but it could easily be normalised to. Such a process is canonical in the Hawkes literature, and can be represented as an unmarked univariate Hawkes process. It was chosen here because the separable uniform mark distribution is maximally entropic, and therefore minimally informative with respect to the ground intensity. This means that the estimates of the multivariate Hawkes representation will be maximally statistically inefficient, and therefore this study provides a heuristic  ``upper bound'' on the error induced by the increased number of parameters necessitated by the multivariate representation. Marks were simulated as draws from a discrete uniform distribution $f\sim\text{Unif}(1,K)$. 

We simulated S=128 realisations on the time window $[0,5056]$. Subsets of each realisation were then observed on twenty increasingly large time windows commensurate with a range of $10^2$ to $10^4$ points. Maximum likelihood parameter estimates were obtained for multivariate Hawkes process representations $\lambda_{\theta}(t,M|\mathcal{H}_t)$ given $K=1,2,\ldots, 6$ component intensities and an exponential kernel matrix. The properties of such estimates are discussed in \cite{ozaki1979maximum}; consistency of MLE for multivariate Hawkes processes is widely conjectured and numerically validated \cite{bowsher2007modelling}. The parameter vector $\theta^*=\{\alpha^*=1,\beta=^*2,\mu^*=1\}$ was specified as $\int_0^\infty \alpha^*\exp(-\beta^* u) du=\frac{1}{2}$, guaranteeing stationarity. %For completeness, the above described process was carried out with other parameter vectors wherein $\mu^*$ and $\alpha^*$ are unchanged but $\beta^*=3,4,5$. Such specifications are also stationary, and are included to provide additional insight with regards to parameter convergence given decreasing exponential decay.
\subsection{Simulation results and interpretation}
Given specified parameter values and a uniform mark distribution,  \textit{ansatz} parameter values are equal to $\bm\theta^*=\{\bm\alpha^*, \bm\beta^*, \bm\mu^*\}$ where the vector $\bm\mu=\mu^* K^{-1}\bm 1_K$, $\bm\alpha=\alpha^* K^{-1}\bm 1_K \bm1_K^{\top}$ and $\bm\beta={\beta^*}\bm 1_K \bm1_K^{\top}$ where $\bm 1_K=(1,1,\ldots, 1)\in\R^K$. Median absolute error (MAE) was then computed for each MLE estimated parameter vector $\bm{\hat\theta}$, defined by $\text{MAE}(\{\bm{\hat\theta}_s\}_{s=1}^S,\bm\theta^*)=\text{median}\{||\bm{\hat\theta}_s-\bm\theta^*||_1\}_{s=1}^S$. %S^{-1}\sum_{s=1}^S ||\bm{\hat\theta}_s-\bm\theta^*||_1$.
\begin{figure}[!h]
\includegraphics[width=1\textwidth]{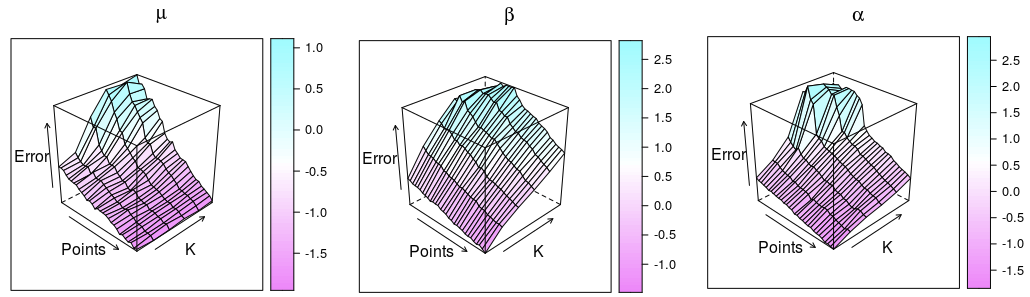}
\caption{Error associated with parameter estimates of multivariate Hawkes representations of univariate marked Hawkes process with an exponential kernel function. Plots show median $L^1$ error of parameter estimates across 128 realisations of size ranging from $10^2$ to $10^4$ points. Error is generally larger as the number of components $K$ increases. $90$\% confidence intervals are shown in appendix \ref{sec:appendixPlots}. Number of points and error are shown on the $\log_{10}$ scale.}
\label{fig:simulation1}
\end{figure}

Multivariate unmarked Hawkes process representations of univariate marked Hawkes processes have more parameters in most cases. For example, given an exponential kernel matrix and constant background parameters, a multivariate Hawkes process has with $K$ components, $K+2K^2$ parameters. A marked univariate Hawkes process with exponential kernel has $3+M$ parameters where $M$ is the dimension of the parameter space of the mark distribution, and in most cases, $K+2K^2>3+M$. However, despite the larger number of parameters, given its flexibility to characterise mark distributions quasi-parametrically, the multivariate representation is a arguably the more parsimonious and computationally tractable model in many scenarios of interest.

A discrete uniform mark distribution was specified because it is maximally entropic (\textit{i.e.} minimally informative with respect to the ground intensity) and therefore, the results visualized in Figure \ref{fig:simulation1} can be used as a rough heuristic bound for worst-case MAE induced by the larger number of parameters necessary to specify a multivariate Hawkes representation of a univariate marked process. In particular, simulated MAE appears $\mathcal{O}\left(K \cdot N(T)^{-\frac{1}{2}}\right)$.

Comparison with the cLSTM neural network architecture \cite{mei2017neural} with regards to the number of parameters required by the multivariate representation is most apt as the cLSTM allows for univariate integer-marked Hawkes process data to be represented by an unmarked parametric multivariate Neural Hawkes process. The cLSTM architecture has the fewest parameters relative to other options such as a fully connected feed forward architecture \cite{omi2019fully} or transformer architecture \cite{zuo2020transformer}, but still has  $d(2K+14d+1)$ parameters given $d$ hidden units and $K$ component intensities. The need for regularization of multivariate Hawkes processes is well discussed in the literature, for example see \cite{bacry2020sparse}, and we leave the ramifications of sparsity inducing regularization in the context of multivariate Hawkes representations of univariate marked Hawkes processes to future work. All code for this study is freely available at \url{https://github.com/ckres213/MVH_Rep}.

\section{Conclusion}\label{sec:conclusion}

In this article we propose a multivariate Hawkes representation for parameterizing univariate marked Hawkes processes. We demonstrate its stationarity, identifiability, and that for a given realisation, the conditional intensity functions of the representative and approximated processes converge in $L^1$ as the number of components increases. We further present a simulation study demonstrating the $L^1$ convergence of parameter estimates to theoretical values given a maximally entropic mark distribution. 

We suggest two possible avenues for future research. Because the density estimator for our process is akin to a histogram estimator, which is clear from Example \ref{thm:PoissProc}, we hypothesise a setup which is analogous to that of a kernel density estimator (see \cite{baddeley2015spatial} for an application to spatial point processes). This would likely increase the ease of application for the multivariate representations, especially in the case of high dimensional mark spaces. Finally, inducing sparsity to reduce the dimension of the parameter space by coalescing subcells of the mark space would likely decrease the model complexity of the representation, thereby increasing its usability.

\begin{appendix}

\section{Proof of Theorem \ref{prop:SeparableHPL1}}\label{sec:L1to0Proof}
\begin{proof}

    We show that for each $T>0$ and $\epsilon>0$ there exists a $K\in \N$, parameter vector $\tilde \theta_K$, and sets $\left\{ A_i\right\}_{1 \leq i \leq K}$ satisfying properties P1-P4, such that 
    \begin{equation}\label{eq:HPApprCompactForm}
        I\vcentcolon=\norm{\lambda_{\text{HP}}(t,M|\HH_t)-\lambda_{\tilde \theta_K}(t,M|\HH_t)}_{L^1([0,T]\times \mathcal M)}<\epsilon.
    \end{equation}

     We first present a relevant lemma.
     \begin{lemma}
Suppose that $(\I_1,\mathcal{F}_1,\Pb_1)$ and $(\I_2,\mathcal{F}_2,\Pb_2)$ are probability spaces and that
\[f_1:\I_1\to L^1(\M), \quad f_2:\I_2 \to L^1(\M), \quad \xi:\I_2\to C(\M)\]
where $f_1,$ $f_2$ and $\xi$ are equicontinuous on $\M$, a compact CSMS, for $\I_1$ or $\I_2$ respectively. Define
\[\bar f_1(M)=\int_{\I_1}f_1(M;\omega_1)d\Pb_1(\omega_1), \quad \overline{f_2\xi}(M,M')=\int_{\I_2}f_2(M;\omega)\xi(M';\omega)d\Pb_2(\omega), \]
then $\bar f_1$ and $\overline{f_2\xi}(M,M')$ are uniformly continuous on $\M$ and $\M^2$ respectively.
     \end{lemma}
\begin{proof}
         Since $\M$ is compact $f_1(\cdot;\omega)$ is uniformly continuous for each $\omega\in \I_i$. Moreover, $\M^2$ is compact so $f_2(M;\omega)\xi(M';\omega)$ is also uniformly continuous for each $\omega\in \I_2$. Hence for every $\epsilon>0$, $\omega \in \I_1$ and every $M \in \M$, there exists a $\delta>0$ so that if $d(M,M')<\delta$ then $| f_1(M,\omega)- f_1(M',\omega)|<\epsilon$. Clearly
         \begin{align*}
         \left|\bar f_1(M)-\bar f_1(M') \right| \leq \int_{\I_1}\left| f_1(M;\omega)-f_1(M';\omega)\right|d\Pb_1(\omega)<\epsilon
         \end{align*}
         since $\Pb$ is a probability measure. This exact argument can be extended to $\overline{f_2\xi}(M,M')$.
\end{proof}

     We first construct the sets $\left\{ A_i\right\}_{1 \leq i \leq K}$ that cover $\M$ and satisfy properties P1-P4. For any $\epsilon>0$, cover $\M$ by the finitely many balls $\mathbb{B}(M_i,\delta)$ such that for any $M,M' \in \mathbb{B}(M_i,\delta)$ and $m,m' \in \mathbb{B}(M_j,\delta)$ we have 
     \[|\bar f_1(M)-\bar f_1(M')|<\epsilon, \quad |\overline{f_2\xi}(M,m)-\overline{f_2\xi}(M',m')|<\epsilon, \quad |g(t,M)-g(t,M')|<\epsilon \quad \forall t.\]
     We now construct
     \[A_i=\left(\mathbb{B}(M_i,\delta) \setminus \bigcap_{j=1}^{i-1}\mathbb{B}(M_j,\delta)\right)\]
     and discard the sets that have measure $0$. Since we assume $\M$ has finite measure, then so too does every $A_i\subset \M$. Clearly, properties P1-P4 are satisfied by these sets.
     
     Now we unpack definitions and notation to find 
    \begin{multline*}
    I=\int_0^T\int_{\mathcal M}\Bigg|\bar f_1(M)\Lambda+ \sum_{\ell:t_\ell<t}g(t-t_\ell;M_\ell)\overline{f_2\xi}(M,M_\ell)-\\ \sum_{i=1}^K \chi_{A_i}(M)\left(\lambda_{0i}+\sum_{j=1}^K\sum_{t_{\ell,j}:t_{\ell,j}<t}\alpha_{ij}g_{ij}(t-t_{\ell,j})\right)\Bigg|d\mu(M)dt.
    \end{multline*}

 By construction $\bigcup_{i=1}^K A_i$ covers $\mathcal M$ and $A_m \cap A_n = \emptyset$ when $m \neq n$ therefore 
 \begin{multline*}
 I=\int_0^T\sum_{i=1}^K \int_{A_i} \Bigg|\bar f_1(M)\Lambda+\sum_{\ell:t_\ell<t}g(t-t_\ell;M_\ell)\overline{f_2\xi}(M,M_\ell)-\\ \left(\lambda_{0i}+\sum_{j=1}^K\sum_{t_{\ell,j}:t_{\ell,j}<t}\alpha_{ij}g_{ij}(t-t_{\ell,j})\right)\Bigg|d\mu(M)dt
 \end{multline*}
 
 \begin{multline*}
 \leq \int_0^T\sum_{i=1}^K \int_{A_i} \left|\bar f_1(M)\Lambda-\lambda_{0i}\right|d\mu(M)dt+\\ \int_0^T\sum_{i=1}^K\int_{A_i}\left|\sum_{\ell:t_\ell<t}g(t-t_\ell;M_\ell)\overline{f_2\xi}(M,M_\ell)- \sum_{j=1}^K\sum_{t_{\ell,j}:t_{\ell,j}<t}\alpha_{ij}g_{ij}(t-t_{\ell,j})\right|d\mu(M)dt.
 \end{multline*}
Denote 
\[\tilde f_{i,1}=\frac{1}{\mu(A_i)}\int_{A_i}\bar f_1(M)d\mu(M),\quad \widetilde{f_2\xi}_{ij}=\frac{1}{\mu(A_i\times A_j)}\int_{A_i\times A_j}\overline{f_2\xi}(M,M')d\mu(M\times M').\]
Define $\lambda_{0i}=\Lambda \tilde f_{i,1}$, $\beta_{ij}\in \beta(A_j)=:\{\beta(M) \in \R:M\in A_j\}$, and pick $\alpha_{ij}=\widetilde{f_2\xi}_{ij}$. Re-index the middle sum to $\sum_{j=1}^K\sum_{t_{\ell,j}:t_{\ell,j}<t}$ which is summing over each event once by type, as opposed to their arrival time. Then 
\begin{multline*}
    I \leq \int_0^T\sum_{i=1}^K \int_{A_i} \Lambda \left|\bar f_1(M)-\tilde f_{i,1}\right|d\mu(M)dt+\\ \int_0^T\sum_{i=1}^K\int_{A_i}\left|\sum_{j=1}^K\sum_{t_{\ell,j}:t_{\ell,j}<t}g(t-t_{\ell,j};M_{\ell,j}) \overline{f_2\xi}(M,M_{\ell,j})- \widetilde{f_2\xi}_{ij}g_{ij}(t-t_{\ell,j}) \right|d\mu(M)dt.
\end{multline*}

\begin{multline}\label{eq:L1Proof2}
   \implies I \leq \int_0^T\sum_{i=1}^K \int_{A_i} \Lambda \left|\bar f_1(M)-\tilde f_{i,1}\right|d\mu(M)dt+\\ \hspace{1.7cm}\int_0^T\sum_{i=1}^K\int_{A_i}\sum_{j=1}^K\sum_{t_{\ell,j}:t_{\ell,j}<t}\overline{f_2\xi}(M,M_{\ell,j})\left|g(t-t_{\ell,j};M_{\ell,j})-g_j(t-t_{\ell,j})\right|d\mu(M)dt+\\
   \int_0^T\sum_{i=1}^K\int_{A_i}\sum_{j=1}^K\sum_{t_{\ell,j}:t_{\ell,j}<t}g_{ij}(t-t_{\ell,j})\left|\overline{f_2\xi}(M,M_{\ell,j})- \widetilde{f_2\xi}_{ij}\right|d\mu(M)dt\\=\vcentcolon\text{LHS}+\text{MT}+\text{RHS}. 
\end{multline}
Consider the first term in Equation \eqref{eq:L1Proof2} (LHS). We change the order of integration and integrate over $t$ first. Since $\bar f_1$ is uniformly continuous $\left|\bar f_1(M)-\tilde f_{i,1}\right|<\epsilon$. Hence, $\text{LHS}<\mu(\M)\Lambda T \epsilon$. 

Before we bound the next term (MT), we define $\bm{f\xi} \vcentcolon=\sup_{(M,M') \in \M^2}\overline{f_2\xi}(M,M_{\ell,j})<\infty$ since $\overline{f_2\xi}(M,M_{\ell,j})$ is continuous and $\M^2$ is compact. Then using the continuity of $g(t;M)$, with respect to $M$,
\begin{align*}
\text{MT} \leq& \int_0^T\sum_{i=1}^K\int_{A_i}\sum_{j=1}^K\sum_{t_{\ell,j}:t_{\ell,j}<t}\bm{f\xi}\left|g(t-t_{\ell,j};M_{\ell,j})-g_j(t-t_{\ell,j})\right|d\mu(M)dt\\
<& \int_0^T\sum_{i=1}^K\int_{A_i}\sum_{j=1}^K\sum_{t_{\ell,j}:t_{\ell,j}<t}\bm{f\xi}\epsilon d\mu(M)dt\\
<&\bm{f\xi}\epsilon N(T)\mu(\M).
\end{align*}
To bound the final term, we first use continuity to find that $\left|\overline{f_2\xi}(M,M_{\ell,j})- \widetilde{f_2\xi}_{ij}\right|<\epsilon.$ We also use that for every $i,j \in \{1,2,\dots K\}$, $g_{ij}(t)$ is a density with respect to $t$. Therefore, 

 \begin{align*}     RHS\leq& \int_0^T\sum_{i=1}^K\int_{A_i}\sum_{j=1}^K\sum_{t_{\ell,j}:t_{\ell,j}<t}g_{ij}(t-t_{\ell,j})\epsilon d\mu(M)dt\\
 \leq& \sum_{i=1}^K\int_{A_i}N(T)\epsilon d\mu(M)=\mu(\M)N(T)\epsilon.
 \end{align*}
Hence, for every $\varepsilon>0$ we can find a $K\in \N$, sets $\left\{A_i\right\}_{1\leq i \leq K}$ and parameter vector $\tilde \theta_K$ such that
\[\norm{\lambda_{\text{HP}}(t,M|\HH_t)-\lambda_{\tilde \theta}(t,M|\HH_t)}_{L^1([0,T]\times \mathcal M)}< \varepsilon.\]
By carefully selecting the parameter vector, we can remove the implicit dependence of $K$ on $T$. This is so that can $K$ can be selected to be arbitrarily large for any fixed $T$ so that Equation \eqref{eq:HPApprCompactForm} holds for every finite positive $T$ (i.e. $I \to 0$ as $K \to \infty$ for each fixed $T>0$). 

    \end{proof}

\section{Proof of Theorem \ref{thm:Stationarity}}\label{sec:proofstationarity}
\begin{proof}
Because the targeted non-separable Hawkes process is stationary, the expected number of offspring from a non-immigrant event is less than $1$ \cite{zhuang2013stability}. Specifically, $I\vcentcolon=\E[\xi(M)|M\text{ has a parent}]<1$ which can be expressed as
\[I=\int_\M \left(\int_{\I_2} f_2(M;\omega) d\Pb_2(\omega)\right) \xi(M)d\mu(M)=1-\sigma<1\]  for $\sigma \in (0,1)$. Since the assumptions of Theorem \ref{prop:SeparableHPL1} hold we use results from its proof, in the special case that $\xi$ is independent of $\omega$. In which case, for any $\epsilon>0$ there exists a $K \in \N$ and set of sets $\{A_i\}_{1\leq i \leq K}$ satisfying P1-P4 such that for any $M,M' \in A_i$,
     \[\quad |\bar f_2(M)-\bar f_2(M')|<\epsilon, \quad |\xi(M)-\xi(M')|<\epsilon\]
     where $\bar f_2(M)=\int_{\I_2}f_2(M;\omega)d\Pb_2(\omega)$, for every $i \in \{1,2,\dots K\}$. Then by picking $\xi_i\in \{\xi(M)\in \R^+:M\in A_i\}$ and $\tilde f_{i,2}=\int_\M \bar f_2(M)d\mu(M)$
     we examine the quantity
\begin{align*}
    &\left|\int_\M \bar f_2(M)\xi(M)d\mu(M)-\sum_{i=1}^K \tilde f_{i,2} \xi_i\mu(A_i) \right|\leq \sum_{i=1}^K \int_{A_i}\left|\bar f_2(M)\xi(M)-\tilde f_{i,2}\xi_i\right|d\mu(M)\\
    &\leq  \sum_{i=1}^K \int_{A_i}\bar f_2(M)\left|\xi(M)-\xi_i\right|d\mu(M)+\sum_{i=1}^K \int_{A_i} \xi_i\left|\bar f_2(M)-\tilde f_{i,2}\right|d\mu(M)\\
    &< \epsilon(1+\bar \xi \mu(\M))=\vcentcolon\varepsilon %First term from uniform continuoity of \xi and second by bounding \xi and using the lemma
\end{align*}
where $\bar \xi\vcentcolon=\sup_{M \in \M}\xi(M)<\infty$ since $\M$ is compact and $\xi$ is continuous. Thus, for every $\varepsilon>0$ there exists a $K \in \N$ such that
\begin{equation}\label{eq:SimpFuncDiag}
   J\vcentcolon= \sum_{i=1}^K \tilde f_{i,2} \xi_i \mu(A_i) <I+\varepsilon=1-\sigma +\varepsilon.
\end{equation}
In particular if $\varepsilon<\sigma$ then $J<1$. 

It is clear that in the case $\xi$ is independent of $\omega$ that the parameter vector $\tilde \theta _K$ has entries of the form $\lambda_{0i}=\Lambda \tilde f_{i,1}$, $\beta_{ij}=\beta_j$, $\alpha_{ij}=\tilde f_{i,2}\xi_j$. Before, presenting the stationarity condition for our multivariate representation we demonstrate another lemma.

\begin{lemma}\label{lem:StatDensity}
    For every $K \in \N$ the mark density of the offspring events induced by $\lambda_{\tilde \theta_K}(t,M|\HH_t)$ is 
    \[f(M)=\sum_{i=1}^K \tilde f_{i,2}\chi_{A_i}(M).\]
\end{lemma}
\begin{proof}
    For any Borel set $A \in \M$, by Bayes' Theorem

\begin{align*}
    &\Pb[m \in A| m \text{ is an offspring}]=\frac{\Pb[m \text{ is an offspring}|m \in A]\Pb[m\in A]}{\Pb[m \text{ is an offspring}]}
    \\=&\frac{\int_A \sum_{i,j=1}^K \chi_{A_i}(M)\sum_{t_{\ell,j}:t_{\ell,j}<t}\alpha_{ij}g_{ij}(t-t_{\ell,j}) d\mu(M)\frac{\int_A\lambda_{\tilde \theta_K}(t,M|\HH_t)d\mu(M)}{\int_\M \lambda_{\tilde \theta_K}(t,M|\HH_t)d\mu(M)}}{\int_A\lambda_{\tilde \theta_K}(t,M|\HH_t)d\mu(M)}\times\\
    &\left(\frac{\int_\M \lambda_{\tilde \theta_K}(t,M|\HH_t)d\mu(M)}{\int_\M \sum_{i,j=1}^K \chi_{A_i}(M)\sum_{t_{\ell,j}:t_{\ell,j}<t}\alpha_{ij}g_{ij}(t-t_{\ell,j}) d\mu(M)} \right)\\
    =&\frac{\int_A \sum_{i,j=1}^K \chi_{A_i}(M)\sum_{t_{\ell,j}:t_{\ell,j}<t}\alpha_{ij}g_{ij}(t-t_{\ell,j}) d\mu(M)}{\int_\M \sum_{i,j=1}^K \chi_{A_i}(M)\sum_{t_{\ell,j}:t_{\ell,j}<t}\alpha_{ij}g_{ij}(t-t_{\ell,j}) d\mu(M)}.
\end{align*}
Clearly, the density $f$ with respect to the measure $\mu$ is then
\[f(M)=\frac{ \sum_{i,j=1}^K \chi_{A_i}(M)\sum_{t_{\ell,j}:t_{\ell,j}<t}\alpha_{ij}g_{ij}(t-t_{\ell,j})}{\int_\M \sum_{i,j=1}^K \chi_{A_i}(M)\sum_{t_{\ell,j}:t_{\ell,j}<t}\alpha_{ij}g_{ij}(t-t_{\ell,j}) d\mu(M)}.\]
Now using the parameter values $\tilde \theta_K$
    
    \begin{align*}
     f(M)=&\frac{\sum_{i=1}^K \chi_{A_i}\tilde f_{i,2}(M)\sum_{j=1}^K\sum_{t_{\ell,j}:t_{\ell,j}<t}\xi_j  g_j(t-t_{\ell,j})}{\sum_{i=1}^K \mu(A_i)\tilde f_{i,2}\sum_{j=1}^K\sum_{t_{\ell,j}:t_{\ell,j}<t}\xi_j g_j(t-t_{\ell,j})}=\sum_{i=1}^K \tilde f_{i,2}\chi_{A_i}(M)
    \end{align*}
since $\sum_{i=1}^K \mu(A_i)\tilde f_{i,2}=1$.
\end{proof}

We now demonstrate that the expected cluster size is finite, hence showing stationarity for the multivariate representation. 

The expected number of first generation offspring of an event of type $j$, labelled $E_j$, is the integral over $[0,T]\times \M$ of its contribution to the intensity. Therefore,
\[E_j=\int_0^T\int_\M \sum_{i=1}^K\chi_{A_i}(M)\alpha_{ij}g(t)d\mu(M)dt=\sum_{i=1}^K \mu(A_i)\alpha_{ij}.\]
Alternatively, the expected number of first generation offspring of an event of size $M$ is
\[E(M)=\sum_{j=1}^K \chi_{A_j}(M)E_j=\sum_{j=1}^K\chi_{A_j}(M)\sum_{i=1}^K \mu(A_i)\alpha_{ij}.\]
The expected number of first generation offspring from a non-immigrant event is then
\begin{align*}
E_C=&\int_\M E(M)f(M)d\mu(M)=\int_\M \sum_{j=1}^K\chi_{A_j}(M)\left(\sum_{i=1}^K \mu(A_i)\alpha_{ij}\right)f(M) d\mu(M)\\
=&\int_\M \sum_{j=1}^K\chi_{A_j}(M)\sum_{i=1}^K \mu(A_i)\alpha_{ij} \tilde f_{j,2} d\mu(M)=\sum_{j=1}^K \mu(A_j)\sum_{i=1}^K \mu(A_i)\tilde f_{i,2} \xi_j   \tilde f_{j,2} \\
=& \sum_{j=1}^K \mu(A_j)\xi_j   \tilde f_{j,2}\sum_{i=1}^K \mu(A_i) \tilde f_{i,2} =\sum_{j=1}^K \mu(A_j)\xi_j   \tilde f_{j,2}=J<1-\sigma +\varepsilon.
\end{align*}
Hence we pick $K$ large enough so that $J<1$, which implies $E_C<1$. Therefore, the mean number of first generation offspring of a non-immigrant event is less than one. It follows that the geometric series equal to the mean number of total offspring from a non-immigrant event, $\sum_{r=0}^\infty E_C^r$, is finite. Immigrant events produce a finite number of offspring events a.s. because
\[\int_\M \bar f_1(M) \xi(M)d\mu(M)<\mu(\M)\sup_{M\in \M }\bar f_1(M)\xi(M)<\infty\]
by the continuity of the integrand as well as the compactness and finite measure of $\M$. Then immigrant events a.s. produce only finitely many first generation offspring, each of which only produce finitely many offspring in total. Therefore, a.s. every cluster is off a finite size. Since branching processes are stationary if their mean cluster size is finite (e.g. Exercise 6.3.5 of \cite{daley2003introduction}%or \cite{zhuang2013stability}
) this concludes the proof. 

\end{proof}

\section{Proof of Theorem \ref{thm:TheoremIdentifiability}}\label{Sec:IdentifiabilityProof}

\begin{proof}
    
Since conditional intensities uniquely define the finite dimensional distributions of the point process \cite{daley2003introduction}, to show identifiability we want to demonstrate that for almost every $(t,M) \in [0,T]\times \mathcal M$ that $\lambda_\theta(t,M|\HH_t)=\lambda_{\theta'}(t,M|\HH_t)  \iff \theta=\theta'$. First of all, it is clear that if $\theta=\theta'$ then $\lambda_\theta(t,M|\HH_t)=\lambda_{\theta'}(t,M|\HH_t)$.

We now assume that $\lambda_\theta(t,M|\HH_t)=\lambda_{\theta'}(t,M|\HH_t)$ almost everywhere (a.e.) on $[0,T] \times \mathcal{M}$. First consider $(t,M) \in [0,t_1]\times A_i$. Since we assume $t_1>0$ we find 
\[\lambda_\theta(t,M|\HH_t)=\lambda_{\theta'}(t,M|\HH_t) \implies \lambda_{0i}=\lambda_{0i}'\]
since the history is empty. $i$ was arbitrary, hence $\lambda_{0i}$ is identifiable for every $i \in \{1,2,\dots K\}$. 

Consider the first event $t_1$ and, by relabelling to ease notation, suppose that it is of type $1$, i.e. $M_1 \in A_{1}$. Then for $(t,M) \in (t_1,t_2]\times A_{i}$ we have 
\begin{align*}\lambda_{\theta}(t,M|\HH_t)=\lambda_{0i}+\alpha_{i1}g_{i1}(t-t_1),\quad \lambda_{\theta'}(t,M|\HH_t)=\lambda_{0i}+\alpha'_{i1}g'_{i1}(t-t_1)
\end{align*}
\begin{align} \lambda_\theta(t,M|\HH_t)=&\lambda_{\theta'}(t,M|\HH_t) \implies\alpha_{i1}g_{i1}(t-t_1)=\alpha'_{i1}g'_{i1}(t-t_1)\\
\iff & \alpha_{i1}g_{i1}(t-t_1)-\alpha'_{i1}g'_{i1}(t-t_1)=0  \mbox{ a.e. in }  (t_1,t_2]\times A_{i}.\label{eq:LinIndep}
\end{align}
 By linear independence Equation \eqref{eq:LinIndep} can only be true if $g_{i1}(t-t_1) = g'_{i1}(t-t_1) \iff \beta_{i1} = \beta'_{i1}$ or $\alpha_{i1}=\alpha'_{i1}=0$.

In the first case, i.e. if $g_{i1}(t-t_1) = g'_{i1}(t-t_1)$ then
\[\alpha_{i1}g_{i1}(t-t_1)=\alpha'_{i1}g'_{i1}(t-t_1) \iff \alpha_{i1}=\alpha'_{i1}\]
since $g_{i1}(t-t_1)>0$ by assumption. 

Otherwise, $\alpha_{i1}=\alpha'_{i1}=0$. This means that $\beta_{i1}$ is unidentifiable, however the probability structure of the process is independent of $\beta_{i1}$ so it is trivially unidentifiable. Hence, for every $i$, when the probability structure of the point process depends on $\alpha_{i1}$ and $\beta_{i1}$ these parameters are identifiable. Hence our inductive base case is true.

Now suppose we have observed events of type $1,2,\dots L$, as in there exists an event $(t_\ell,M_\ell)$ such that $M_\ell \in A_{n}$ for every $n \in \{1,2,\dots L\}$, which we know to be true by assumption A4. 

We inductively assume $\alpha_{in}=\alpha'_{in}$ and $\beta_{in}=\beta'_{in}$ for $L<K$ and every $i \in \{1,2,\dots K\}$. We now show that $\alpha_{i(L+1)}=\alpha'_{i(L+1)}$ and $\beta_{i(L+1)}=\beta'_{i(L+1)}$.

Suppose that $t_m$ is the first event of type $L+1$ and that $\mathcal K_t$ is the index set of every type of event that has occurred prior to time $t$. Consider, $(t,M) \in (t_m,t_{m+1}]\times A_i$ where $\lambda_\theta(t,M|\HH_t)=\lambda_{\theta'}(t,M|\HH_t)$ a.e., then
\begin{equation}\label{eq:InducHyp1}
    \lambda_{0i}+\sum_{k \in \mathcal K_t}\sum_{t_{\ell,k}:t_{\ell,k}<t}\alpha_{ik}g_{ik}(t-t_{\ell,k})=\lambda'_{0i}+\sum_{k \in \mathcal K_t}\sum_{t_{\ell,k}:t_{\ell,k}<t}\alpha'_{ik}g'_{ik}(t-t_{\ell,k}) \mbox{ a.e.}.
\end{equation}
However, we have shown $\lambda_{0i}=\lambda'_{0i}$ and by our inductive hypothesis for every $k \in \mathcal K_{t_m}=\{1,2,\dots L\}$, and $i \in \{1,2,\dots K\}$, $\alpha_{ik}=\alpha'_{ik}$ and $\beta_{ik}=\beta'_{ik}$. By cancelling these common terms from both sides of Equation \eqref{eq:InducHyp1} we find 
\begin{multline}\label{eq:LinIndep2}
\alpha_{i(L+1)}g_{i(L+1)}(t-t_m)=\alpha'_{i(L+1)}g'_{i(L+1)}(t-t_m)\\ \iff \alpha_{i(L+1)}g_{i(L+1)}(t-t_m)-\alpha'_{i(L+1)}g'_{i(L+1)}(t-t_m)=0.\end{multline}
The right hand side of Equation \eqref{eq:LinIndep2} is identical to Equation \eqref{eq:LinIndep} except one must replace $1$ by $L+1$. As in this previous case, we find  $\alpha_{i(L+1)}=\alpha'_{i(L+1)}$ and $\beta_{i(L+1)}=\beta'_{i(L+1)}$ when the probability structure of the process depends on these parameters. Therefore, if the $L^{th}$ case is true then the $(L+1)^{th}$ case is true, i.e. $\alpha_{in}=\alpha'_{in}$ and $\beta_{in}=\beta'_{in}$ for $n \in \mathcal K_{t_{M+1}}=\{1,2,\dots L+1\}$ and every $i \in \{1,2,\dots K\}$.  By the principle of mathematical induction we find that $\lambda_\theta(t,M|\HH_t)=\lambda_{\theta'}(t,M|\HH_t) \mbox{ a.e. } \implies \theta=\theta'$. 

Therefore, $\lambda_\theta(t,M|\HH_t)=\lambda_{\theta'}(t,M|\HH_t) \mbox{ a.e. } \iff \theta=\theta'$. Equivalently, all parameters that the probability structure of the process depends on are identifiable.

\end{proof}

\newpage
\section{Figures of Section \ref{Sec:Simulation}}\label{sec:appendixPlots}
\begin{figure}[!h]
\centering
  \includegraphics[width=1\textwidth]{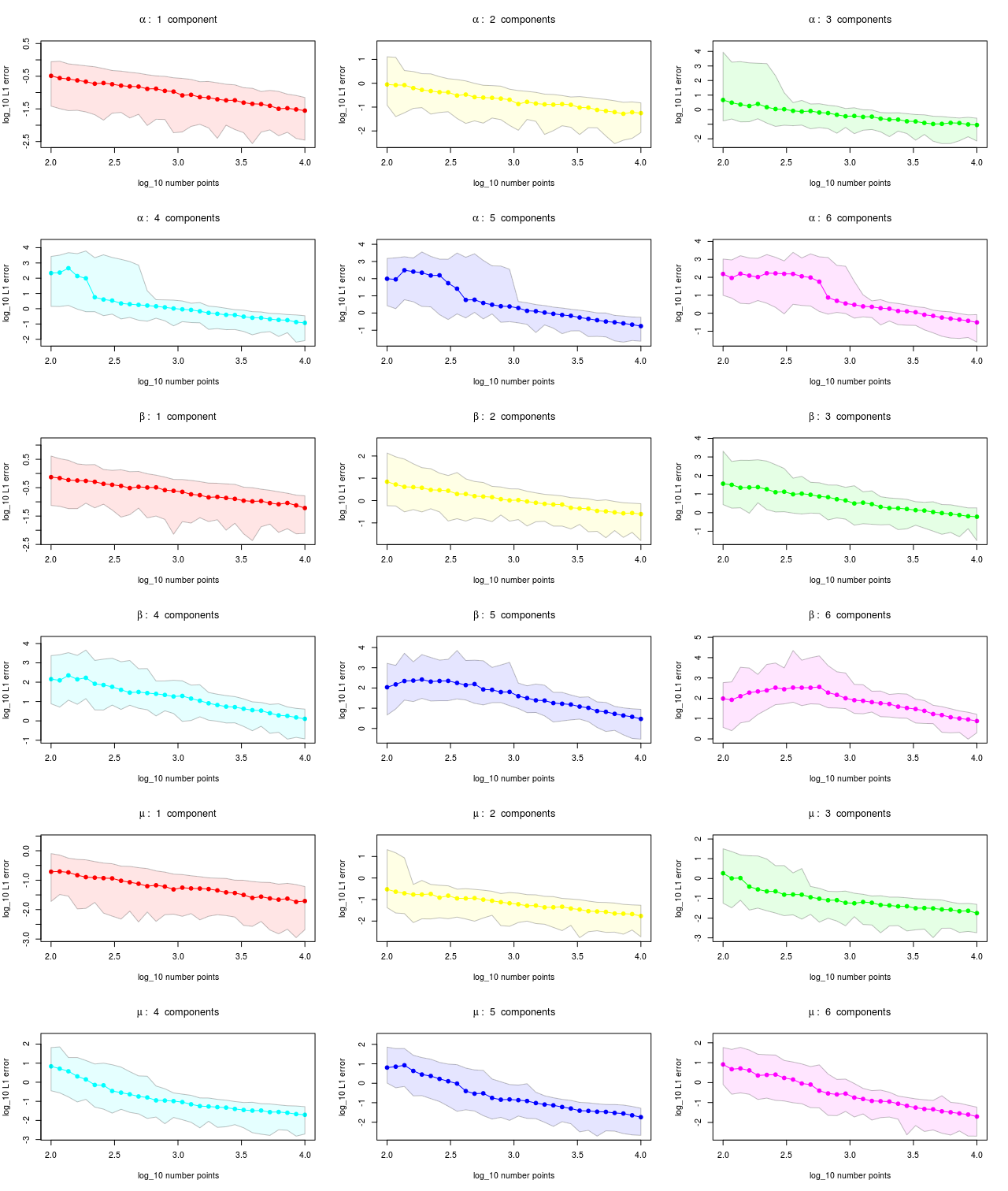}
\caption{Approximate 90\% confidence interval for median $L^1$ error associated with the parameters of the multivariate Hawkes representation of the univariate Hawkes process specified in Equation \ref{eqn:sim_uni}. Median and associated confidence interval based on 128 realisations.}
\label{fig:alphaCI}
\end{figure}
\section{Acknowledgements and Funding}
    The authors wish to acknowledge the use of New Zealand eScience Infrastructure (NeSI) high performance computing facilities, and consulting support as part of this research. New Zealand's national facilities are provided by NeSI and funded jointly by NeSI's collaborator institutions and through the Ministry of Business, Innovation \& Employment's Research Infrastructure programme.  \url{https://www.nesi.org.nz}.

This project was supported by the Royal
Society of New Zealand Marsden Fund (contact MFP-UOO2323).

\end{appendix}

\bibliographystyle{abbrv} % Style BST file (imsart-number.bst or imsart-nameyear.bst)

\bibliography{bibliography.bib}

\end{document}